\documentclass[final,11pt]{article}

\usepackage{amsmath}
\usepackage{amssymb}
\usepackage{latexsym}
\usepackage[utf8]{inputenc}
\usepackage[T1]{fontenc}
\usepackage[pdftex,final]{graphicx}
\usepackage[usenames,dvipsnames]{xcolor}
\usepackage{ifdraft}%
\usepackage{url}
\usepackage{xspace}
\usepackage{times}
\usepackage{stmaryrd}
\usepackage{verbatim}
\usepackage{hyperref}
\usepackage{proof}
\usepackage{longtable}
\usepackage{enumerate}
\usepackage{authblk}
\hypersetup{
  pdfauthor = {Viviana Bono and Marcin Benke and Aleksy Schubert},
  pdftitle = {Lucretia --- a type system for objects in languages with reflection},
  pdfsubject = {F.3.3, F.4.1.},
  pdfkeywords = {????},
  pdfcreator = {LaTeX with hyperref package},
  pdfproducer = {pdflatex},
  bookmarksdepth=3,
  bookmarks=true,
  final
}

\newcommand{\cbool}{\mathit{bool}\xspace}

\newcommand{\kfunc}{\mathsf{func}\xspace}
\newcommand{\kbreak}{\mathsf{break}\xspace}
\newcommand{\knew}{\mathsf{new}\xspace}

\newcommand{\ktrue}{\mathsf{true}\xspace}
\newcommand{\kfalse}{\mathsf{false}\xspace}
\newcommand{\klet}{\mathsf{let}\xspace}
\newcommand{\kif}{\mathsf{if}\xspace}
\newcommand{\letin}[3]{\klet\;\xspace #1 = #2\xspace\;\mathsf{in}\; %
                       \xspace #3\xspace}
\newcommand{\ite}[3]{\kif\;\xspace (#1)\; %
                     \mathsf{then}\;\xspace #2 \;\mathsf{else}\;\xspace #3\xspace}
\newcommand{\iha}[4]{\mathsf{ifhasattr}\; (#1,#2)\; %
                     \mathsf{then}\;\xspace #3 \;\mathsf{else}\;\xspace #4\xspace}
\newcommand{\breakle}[2]{\kbreak\;\xspace #1\;\xspace #2\xspace}
\newcommand{\dom}[1]{\mathrm{dom}(#1)\xspace}

\newcommand{\vdc}{\preceq_{\mathit{c}}}
\newcommand{\atype}[4]{[{#2};{#1}]\implies[{#3};{#4}]}
\newcommand{\matches}{\mathrel{\mathord{<}\!\raisebox{1pt}{\scriptsize $\#$}}}

\newcommand{\TConst}{\mathsf{Const}_T\xspace}
\newcommand{\TVars}{\mathcal{V}_T\xspace}
\newcommand{\Types}{\mathsf{Types}\xspace}
\newcommand{\Constr}{\mathsf{Constr}\xspace}
\newcommand{\updconstr}{\hookleftarrow}
\newcommand{\econtext}[1]{\vec{E}\left\langle#1\right\rangle}
\newcommand{\draft}[1]{\ifdraft{{{\color{red} [[[#1]]]}}}{} }
\newcommand{\draftmargin}[1]{\ifdraft{\marginpar{{\color{red} [[[#1]]]}}}{} }

\newcommand{\vivi}[1]{\ifdraft{{{\color{JungleGreen} #1}}}{}}

\let\vd=\vdash
\newtheorem{definition}{Definition}[section]
\newtheorem{lemma}[definition]{Lemma}
\newtheorem{theorem}[definition]{Theorem}
\newtheorem{proposition}[definition]{Proposition}
\def\sqr#1#2{\vbox
 {\hrule height#2
  \hbox{\vrule width#2 height#1 \kern#1 \vrule width#2}%
  \hrule height#2}}
\def\squarek{\sqr{1.5ex}{.4pt}}
\def\koniec{\mbox{}
\nolinebreak[4]{\mbox{}\hspace*{20pt}\hfill\squarek}}
\newenvironment{proof}{\vspace*{-0.8em}\noindent{\bf
   Proof:\\}}{}

\title{Lucretia\,---\,a type system for objects in languages with reflection%
  \thanks{See \url{http://www.youtube.com/watch?v=IuezNswtRfo}}}
\author[1]{Viviana Bono}
\author[2]{Marcin Benke}
\author[2]{Aleksy Schubert}
\affil[1]{Dipartimento di Informatica dell'Università di Torino\thanks{This work was partly supported by the PRIN 12008 DISCO grant.}}
\affil[2]{University of Warsaw\thanks{This work was partly supported by the Polish government grant no N N206 355836.}}
\date{May 1, 2012}

\makeatletter

\begin{document}

\maketitle

\begin{abstract}
  Object-oriented scripting languages such as JavaScript or Python gain
  in popularity due to their flexibility. Still, the growing code bases
  written in the languages call for methods that make possible to automatically
  control the properties of the programs that ensure their stability in the
  running time. We propose a type system, called Lucretia, that makes
  possible to control the object structure of languages with reflection.
  Subject reduction and soundness of the type system with respect to the
  semantics of the language is proved.
\end{abstract}

\section{Introduction}

Scripting object-oriented languages such as JavaScript, Python, Perl or Ruby
became very popular \cite{Tiobe12} due to their succinctness of program
expressions and flexibility of idioms \cite{WrigstadNV09}.
These languages optimise the programmer time, rather than the machine time and
are very effective when small programs are constructed
\cite{Prechelt00,WrigstadEFNV09}. The advantages of the languages that help when
a short programs are developed can be detrimental in case big applications are
created. Short code that has clear advantages when small programs are
constructed \cite{Ousterhout98} can provide less information for a person that
looks for clues on what particualr piece of code is doing (and this is one of
the most frequent activities during software maintenance tasks, see e.g.
\cite{Sasso96,KoMCA06}). Moreover, the strong invariants a programmer can rely
on in case of statically typed languages are no longer valid as type of a
particular value can change with each instruction and mostly in an uncontrolled
way with each function call in the program.

Still, systems that handle complex and critical tasks such as Swedish premium
pension system \cite{Stephenson01} are deployed and continuously maintained.
Therefore, it is desirable to study methods which help programmers in
understanding their code and keeping it consistent. Therefore we propose
a type system that handles dynamic features of scripting languages and can
help in understanding of the existing code. 

\draftmargin{alx: I know how to prove the undecidability, but do we need the
proof? If so I can type it down.}
The type system we propose is presented in the style where many type annotations
must be present to guarantee correct typechecking. In particular, we require
that function declarations as well as labelled instructions are decorated with
their types. What is more, the type system is so strong that the reconstruction
of the types is (most probably) undecidable.  However, the system is designed in
such a way that it will primarily be used on top of a type inference algorithm
and admits a~wide range of type inference heuristics that are not complete, but
provide correct type annotations for wide range of programs. In addition,
we would like to encourage programmers to add type annotations to their
programs since the annotations can serve as an important documentation of
the code invariants assumed by developers in the program construction phase.

The inherent feature of the scripting languages is that a running time type of a
particular variable may change in the course of program execution. This problem
can be solved to some extent through introduction of single assignment form for
local variables. Still, this cannot be applied to object fields. The natural
semantics of programs is such that the fields change and the efforts to
circumvent this principle will most probably result in a complicated solution.
Therefore, the statement that statically describes evolution of the running time
type of a variable cannot be just a type name, but must reflect the journey of
the running time type throughout the control flow graf of the program. It would
be very inconvenient to repeat the structure of the whole control flow graf at
each variable in the program. Therefore, it makes more sense to describe the
type of each visible variable at each statically available program point.
This is the approach we follow in this paper and present a type system which is
inspired by the works on type-and-effect systems \cite{Marino09,Gifford86}.
However, we present our typing judgements in a slightly different manner, i.e.
one where the effect is described by a pair of constraints that express the
update performed by a particular instruction. In a sense the pair together with
the typed expression can be viewed as a triple in a variant of Hoare logic. The
constraints we use to express the type information for an object are
approximations of the actual types by the matching constraints similar to the
ones in \cite{BonoB99}.

\newif\ifevolution
\evolutionfalse

This paper is structured as follows. 
\ifevolution
We present an overview of object 
evolution in scripting languages in Section~\ref{sec:evolution}.
\fi
The language in which we model the scripting languages is presented in
Section~\ref{sec:semantics-and-types}. The type system is presented in
Section~\ref{sec:type-system}. Its properties are demonstrated in
Section~\ref{sec:properties-of-the-type-system}. At the end we present
related work in Section~\ref{sec:related-work} and conclude in
Section~\ref{sec:conclusions}.

\ifevolution
\section{Object evolution in scripting languages}
\label{sec:evolution}

\draft{write down using \cite{LebresneROWV09}}

\paragraph{some random examples below}

what is the definition of [x := l] and other constructs, 

\begin{verbatim}
 y := 1
 if y then
   y := "foo"
 else
   y := 12
 f(y)

Problem with translation to let-representation because there are two

This in SSA is
 y1 := 1
 if y1 then
   y2 := "foo"
 else
   y3 := 12
 y4 = \phi(y2, y3)
 f(y4)

another form of translation from real Python
 let local = new {} in
   let _ = local.y = 1 in
     let _ = if local.y then
               local.y = "foo"
             else
               local.y = 12
     in
       f(local.y)


\end{verbatim}

Note that it looks like we can avoid $\top$, at least to be able to specify partially the content of an object. Actually, in this version of the system there is no way of specifying the actual complete type of an object, which is always a type variable constrained in a $\Phi$. We should be able to write:
\begin{verbatim}
 let local = new  in
  let _ = local.c = 1 in
    let _ = local.d = local in 
     ...
\end{verbatim}
where $d$ should have the type $\texttt{d}:X$, with $X <: \{\texttt{c}:int, \texttt{d}:X\}$. \vivi{To be checked}

Can we do something like:
\begin{verbatim}
let local = new  in
  let _ = local.c = 1 in
    let _ = local.d = local in 
     let _ = local.c = local.c + 1 in
     ...   
\end{verbatim}
that is, making the $\texttt{local}$ id play the role of self? I think it should be possible, if local has type
$X$... But then we need two axioms:

$$\infer{\Gamma(x:t)\vdash x:t ; \Phi}{t \mbox{ not a type variable}}$$

$$\infer{\Gamma(x:X)\vdash x:X ; \Phi, X<:\{\overline{c:t}\}}{}$$

\fi

\section{Model language and its semantics}
\label{sec:semantics-and-types}

The syntax of expressions we work with is presented in Fig~\ref{fig:syntax}.
The figure presents not only the raw syntax of programs, but also the
syntax of evaluation contexts that represent ifnormation on the particular
point in the program the semantical reduction rules elaborate. Moreover,
we present here the full syntax of the expressions that may show up in the
evolution of the expression being reduced. This is reflected by the fact that
we permit locations to occur in the expressions while the intent is they are
not directly visible in programs in the source form.

\begin{definition}[Types]
\label{def:types}
  The type information associated with an expression of the language
  we consider here is combined of two items. One is a representation
  of the actual type and the second is the constraint expression which
  approximates the shape of the type. These components are generated
  with help of the following grammar:
$$
  \begin{array}{lcl}
    t_b    &::=&  c \mid X \mid t_{b,1}\lor t_{b,2} \mid
    \atype{\Psi_1}{t_1,\ldots,t_n}{t_{n+1}}{\Psi_2}\\
    t      &::=&  t_b \mid t_1\lor t_2 \mid  \bot\\
    t_r    &::=&  \{ \overline{l:t} \} \mid \{\}\\
    \Psi   &::=&  X\matches t_r, \Psi \mid \emptyset
  \end{array}
$$
where $c$ is a type constant from $\TConst$, $X\in\TVars$ is a type
variable.  The set $\Types$ is the set of elements generated from $t$,
the set $\Types_b$ is the set of elements generated from $t_b$, the
set $\Types_r$ is the set of elements generated from $t_r$, and the
set $\Constr$ is the set of elements generated from $\Psi$.
\end{definition}

\begin{figure}
\begin{displaymath}
\begin{array}{p{100pt}l@{~}l@{~}l}
  identifiers & x & ::= & \mbox{ \ldots(identifiers)\ldots }\\
  labels      & n      & ::= & \mbox{ \ldots(identifiers)\ldots }\\
  locations   & l \\
  constants   & c      & ::= & \ldots\\
  function value      & v_f    & ::= & \kfunc(x_1,\cdots,x_n).t\; \{ e \}\\
  values      & v      & ::= & x \mid c \mid v_f \mid l\\
  objects     & o      & ::= & \{ \} \mid \{ L_f \}\\
  fields list & L_f    & ::= & n : v \mid n : v, L\\
  function expression & e_f    & ::= & x \mid v_f \\
  expressions & e      & ::= & v \mid \letin{x}{e_1}{e_2}
                                 \mid e_f(e_1,\cdots, e_n)\\
              &        &     &   \mid op_n(e_1,\cdots, e_n)\\
              &        &     &   \mid \ite{e_1}{e_2}{e_3}\\
              &        &     &   \mid \iha{x}{a}{e_1}{e_2}\\
              &        &     &   \mid \breakle{n}{e}
                                 \mid n.t\; \{ e \}\\
              &        &     &   \mid \knew
                                 \mid x_1.x_2 = v \mid x_1.x_2
                                 \\
  basic
  evaluation
  contexts    & E      & ::= & \bullet
                                 \mid \letin{x}{E}{e} \\
              &        &     &   \mid v_f(v_1,\cdots, v_{i-1}, E, e_{i+1},\cdots,e_n)\\
              &        &     &   \mid op_n(v_1,\cdots, v_{i-1}, E,e_{i+1},\cdots,e_n)\\
              &        &     &   \mid \ite{E}{e_2}{e_3}\\
              &        &     &   \mid \breakle{n}{E}\\
              &        &     &   \mid n.t\; \{ E \}\\
  stores      & \sigma & ::= & \cdot \mid (l, o) \sigma\\
\end{array}
\end{displaymath}
\caption{Syntax of $\lambda_M$}
\label{fig:syntax}
\end{figure}

\begin{figure}
{\small
\begin{displaymath}
\begin{array}{l@{~~}l}
  (\mbox{Let-Enter}) & \sigma, \vec{E}\langle \letin{x}{e_1}{e_2}\rangle \leadsto
                         \sigma, \vec{E}\langle \letin{x}{\langle e_1\rangle}%
                                                      {e_2}\rangle\\
  (\mbox{Let})       & \sigma, \vec{E}\langle \letin{x}{\langle v\rangle}
                                                      {e}\rangle \leadsto
                         \sigma, \vec{E}\langle e[x/v]\rangle\\[1.5em]

  (\mbox{Op-Enter})     & \sigma, \vec{E}\langle op_n(e_1,\cdots, e_k)\rangle \leadsto\\
                       & \qquad
                         \sigma, \vec{E}\langle op_n(\langle e_1\rangle,\cdots, e_k)\rangle\\
  (\mbox{Op-Prop})      & \sigma, \vec{E}\langle op_n(v_1,\cdots, v_{i-1}, \langle v_i\rangle, e_{i+1},\cdots, e_k)\rangle \leadsto\\
                       & \qquad
                         \sigma, \vec{E}\langle op_n(v_1,\cdots, v_{i-1}, v_i, \langle e_{i+1}\rangle,\cdots, v_k)\rangle\\
  (\mbox{Op-Back})      & \sigma, \vec{E}\langle op_n(v_1,\cdots,\langle v_k\rangle)\rangle \leadsto\\
                       & \qquad
                         \sigma, \vec{E}\langle op_n(v_1,\cdots, v_k)\rangle\\
  (\mbox{Op-Eval})      & \sigma, \vec{E}\langle op_n(v_1,\cdots, v_k)\rangle \leadsto
                         \sigma, \vec{E}\langle \delta_n(op_n, v_1,\cdots,v_k)\rangle\\[1.5em]

  (\beta\mbox{-Enter}) & \sigma, \vec{E}\langle \kfunc(x_1,\cdots, x_n).t\; \{ e \}(e_1,\cdots, e_k)\rangle \leadsto\\
                       & \qquad
                         \sigma, \vec{E}\langle \kfunc(x_1,\cdots, x_n).t\; \{ e \}(\langle e_1\rangle,\cdots, e_k)\rangle\\
  (\beta\mbox{-Prop})  & \sigma, \vec{E}\langle \kfunc(x_1,\cdots, x_n).t\; \{ e \}(v_1,\cdots, v_{i-1}, \langle v_i\rangle, e_{i+1},\cdots,e_n)\rangle \leadsto\\
                       & \qquad
                         \sigma, \vec{E}\langle \kfunc(x_1,\cdots, x_n).t\; \{ e \}(v_1,\cdots, v_{i-1}, v_i,\langle e_{i+1}\rangle,\cdots,e_n)\rangle\\
  (\beta v)            & \sigma, \vec{E}\langle \kfunc(x_1,\cdots, x_n).t\; \{ e \}(v_1,\cdots,\langle v_k\rangle)\rangle \leadsto\\
                       & \qquad
                         \sigma, \vec{E}\langle e[x_1/v_1,\cdots,x_n/v_n]\rangle\\[1.5em]

  (\mbox{If-Enter})  & \sigma, \vec{E}\langle \ite{e_1}{e_2}{e_3}\rangle \leadsto
                         \sigma, \vec{E}\langle \ite{\langle e_1\rangle}{e_2}{e_3}\rangle\\
  (\mbox{If-True})   & \sigma, \vec{E}\langle \ite{\langle \ktrue\rangle}{e_2}{e_3}\rangle \leadsto
                         \sigma, \vec{E}\langle e_2\rangle\\
  (\mbox{If-False})   & \sigma, \vec{E}\langle \ite{\langle
  \kfalse\rangle}{e_2}{e_3}\rangle \leadsto \sigma, \vec{E}\langle e_3\rangle\\[1.5em]

  (\mbox{Ifhtr-True} )  & \sigma, \vec{E}\langle \iha{l}{a}{e_1}{e_2}\rangle \leadsto
                          \sigma, \vec{E}\langle e_1\rangle\\
                        & \quad \mbox{ when } a\in\dom{\sigma(l)}\\
  (\mbox{Ifhtr-False} ) & \sigma, \vec{E}\langle \iha{l}{a}{e_1}{e_2}\rangle \leadsto
                          \sigma, \vec{E}\langle e_2\rangle\\
                        & \quad \mbox{ when } a\not\in\dom{\sigma(l)}\\[1.5em]

  (\mbox{Brk-Enter})   & \sigma, \vec{E}\langle n.t\;\{ e \}\rangle \leadsto
                         \sigma, \vec{E}\langle n.t\;\{\langle e\rangle \}\rangle\\
  (\mbox{Brk-P})       & \sigma, \vec{E}_1\langle n.t\;\{
                         \vec{E}_2\langle\breakle{n}{v} \rangle\}\rangle\leadsto
                         \sigma, \vec{E}_1\langle v\rangle,\\
                       & \quad \mbox{ when $n$ does not occur as label in }  \vec{E}_2\\
  (\mbox{Lbl-Pop})     & \sigma, \vec{E}\langle n.t\;\{ \langle v\rangle \}\rangle \leadsto
                         \sigma, \vec{E}\langle v\rangle\\[1.5em]

  (\mbox{New})         & \sigma, \vec{E}\langle\knew\rangle \leadsto
                         (l, \emptyset)\sigma, \vec{E}\langle l\rangle \quad l \mbox{
                         fresh }\\
  (\mbox{SetRef})      & \sigma, \vec{E}\langle l.f=v\rangle \leadsto
                         \sigma [l/\sigma(l)[f/v]], \vec{E}\langle v\rangle\\
                         
  (\mbox{Deref})       & \sigma, \vec{E}\langle l.f\rangle \leadsto
                         \sigma, \vec{E}\langle \sigma(l)(f)\rangle\\
                       & \quad \mbox{ when } f\in\dom{\sigma(l)}\\

\end{array}
\end{displaymath}}

\caption{Semantic rules of $\lambda_M$}
\label{fig:semantics}
\end{figure}

The set of objects ${\cal O}$ is the set of partial functions from
identifiers to values.
The set of stores $H$ is the set of finite partial functions from the
set of locations to the set of objects. The set of evaluation contexts
${\cal EC}$ is the set of sequences of {\em basic evaluation contexts}
as presented in Figure~\ref{fig:syntax}. The set of expressions
${\cal EX}$ is the set of expressions as given in
Figure~\ref{fig:syntax}.
The semantics is defined by the relation $\leadsto$ which relates triples
$(\sigma, E, e)\in H\times{\cal EC}\times{\cal EX}$. Let $s$ be
a (partial) function we write $s[x/v]$ to denote function such that  $s[x/v](x)=v$
and $s[x/v](y)=s(y)$ for $y\not=x$ and $y\in\dom{s}$.

When an evaluation context is a sequence $E_1,E_2,\ldots,E_n$, we write it as
\begin{displaymath}
  E_1\langle E_2\langle\ldots E_n\langle\bullet\rangle\ldots\rangle\rangle.
\end{displaymath}
When the relation holds between triples $(\sigma, E, e)$,
$(\sigma', E', e')$ we write the relation as
$\sigma, E\langle e\rangle\leadsto \sigma', E'\langle e'\rangle$. The precise
semantical rules are given in Figure~\ref{fig:semantics}

We do not model inheritance and multiple inheritance of the scripting languages.
We assume that this can be viewed as a notational shortcut for direct
presentation of objects and therefore we cannot model all
features of the class models (e.g. the method updates in classes in Python).

\paragraph{Constructs of the language}

The language we propose has object-oriented features such as object creation
($\knew$), field reference (the dot notation $x_1.x_2$),  and field
modification ($x_1.x_2 = v$). In addition we have object introspetion operation
available through $\iha{\cdot}{\cdot}{\cdot}{\cdot}$ construct. The flow
of information is controlled by $\letin{x}{e_1}{e_2}$ expression 
that in addition to creation of a local variable $x$ makes possible to 
execute $e_1$ and $e_2$ in sequence. The traditional split of control flow
depending on a computed condition is realised by $\ite{\cdot}{\cdot}{\cdot}$.
The exceptional flow can be realised by our labelled instructions combined
with $\kbreak$ statements. At last loops can be organised by recursively
defined functions.

Notably we use types in labelled instructions and in function declarations.
These types can in fact be omitted since they do not play part in the respective
reduction rules. The possible errors that may result by execution of an
operation on a value that the operation is not prepared to take are present
in our system through appropriate definition of $\delta_n$ that
can check the types of the values supplied in the arguments.

\section{The type system}
\label{sec:type-system}

Typing judgements are of the form $\Psi_1;\Gamma;\Sigma\vd e:t_b;\Psi_2$ where
$\Gamma$ is an environment (i.e. mapping from variables to $t_b$),
$\Sigma$ is a mapping from locations to type variables $\TVars$,
$\Psi_i$ for $i=1,2$ is a set of constraints of the form $X\matches t$.

The intended meaning of such a judgment is: evaluating $e$ in the environment
$\Gamma$ with types of locations that match the mapping $\Sigma$ and with the
store satisfying $\Psi_1$ leads to a value of type $t_b$ and store satisfying
$\Psi_2$.
\par\medskip
\paragraph{Record update}

The information about the record being updated is represented by a constraint
on a variable $X$. The rule for record update comes in two variants, depending
on whether the constraint mentions the field being updated.

\newcommand{\uoldrule}{\mbox{\small\rm (update-old)}\xspace}
\newcommand{\ufreshrule}{\mbox{\small\rm (update-fresh)}\xspace}
\begin{displaymath}
  \infer[\uoldrule]%
    {\Psi_1;\Gamma;\Sigma    \vd    x.a=e
       :  t_2;      \Psi_2,  X\matches\{a:t_2,\overline{l:s}\}
    }
    {
      \begin{array}{c}
        \Psi_1; \Gamma;\Sigma   \vd   x  
          :  X;      \Psi_1
      \\
        \Psi_1; \Gamma;\Sigma   \vd   e
          :  t_2;    \Psi_2,  X\matches \{a:t'_2,\overline{l:s}\}\\
      \end{array}
    }
\end{displaymath}
If (we know that) the record contains the field $a$ that is to be updated,
we forget the old value, hence its type ($t'_2$) is ignored; the store after
the update reflects the new value of~$a$.
\begin{displaymath}
  \infer[\ufreshrule]%
    {\Psi_1;\Gamma;\Sigma    \vd     x.a=e
       : t_2;      \Psi_2,  X\matches\{a:t_2,\overline{l:s}\}}
    {
      \begin{array}{c}
        \Psi_1;\Gamma;\Sigma   \vd   x
           : X;\Psi_1\\
        \Psi_1;\Gamma;\Sigma   \vd   e
           : t_2;\Psi_2,X\matches \{\overline{l:s}\}
      \qquad
        a\not\in\overline{l}
      \end{array}
    }
\end{displaymath}
otherwise the constraint for $X$ is amended with  information about the
new value of $a$.

\paragraph{Record access}

\newcommand{\raccrule}{\mbox{\small\rm (access)}\xspace}
\newcommand{\raccmrule}{\mbox{\small\rm (access-m)}\xspace}
We can access a field $a$ provided the constraint on the record guarantees that
it has the field $a$ and it is of a definite type (it cannot be $\bot$).
\begin{displaymath}
  \infer[\raccrule]%
    {
      \Psi;\Gamma,x:X;\Sigma   \vd   x.a : u;\Psi
    }
    {
      \vd u
    &
      \Psi\ni X\matches\{a:u,\overline{l:t}\}
    }
\end{displaymath}

A type is definite ($\vd u$) if it is a type constant, type variable, function
type or a disjunction of definite types. This is defined inductively as
follows:

\begin{displaymath}
\begin{array}{c@{~~~\qquad~~}c}
  \infer{\vd c}{}
&
  \infer{\vd X}{}
\\[2ex]
  \infer{\vd t_1\lor t_2}
        {\vd t_1  &   \vd t_2}
&
  \infer{\vd \atype{\Psi_1}{t_1,\ldots,t_n}{u}{\Psi_2}}
        {}
\end{array}
\end{displaymath}

\paragraph{Conditional instruction}

\newcommand{\condrule}{\mbox{\small\rm (if)}\xspace}
The basic way to split the processing depending on some value is to use
a conditional instruction. This instruction is typed in our system in the
following way.

\begin{displaymath}
  \infer[\condrule]%
  {
    \Psi;\Gamma;\Sigma   \vd   \ite{e_1}{e_2}{e_3}
      : u;  \Psi_2 
  }%
  {
    {\Psi;\Gamma;\Sigma   \vd   e_1
      : \cbool;\Psi_1}
  &
    \Psi_1;\Gamma;\Sigma \vd   e_2
      : u;\Psi_2
  &
    \Psi_1;\Gamma;\Sigma \vd   e_3
      : u;\Psi_2
  }
\end{displaymath}

The typing of the branches need to be typable into the same type.
However, the constraints in each of them can be in principle
different. Therefore we need a rule to weaken the constraint so
that they result in the same set. For this we need an operation that
merges type constraint from two branches. This is defined in the following
way.
\begin{itemize}
\item
$
  \Psi,X\matches\{\overline{l:t}\}\uplus \Psi',X\matches\{\overline{l':t'}\} =
  \Psi\uplus\Psi',X\matches\{\overline{l:t}\}\uplus\{\overline{l':t'}\}
$
\item
$
  \Psi\uplus \emptyset =
  \Psi
$
\item
$
  \emptyset\uplus \Psi =
  \Psi
$
\item
$
  \{l_0:t_0,\overline{l:t}\}\uplus\{l_0:t_0',\overline{l':t'}\} =
  \{l_0:t_0\lor t_0',\overline{l'':t''}\}
$ where $\{\overline{l'':t''}\} = \{\overline{l:t}\}\uplus\{\overline{l':t'}\}$
\item
$
  \{\}\uplus\{\overline{l':t'}\} =
  \{\overline{l':t'\lor\bot}\}
$
\item
$
  \{\overline{l':t'}\}\uplus\{\} =
  \{\overline{l':t'\lor\bot}\}
$
\end{itemize}

\newcommand{\weakenruleb}{\mbox{\small\rm (weaken-$\uplus$)}\xspace}
This combination of two sets of constraints need to be incorporated into
type derivation and this can be done using the following weakening rule.
\begin{displaymath}
  \infer[\weakenruleb]%
  {
    \Psi;\Gamma;\Sigma\vd e:t;\Psi_1\uplus\Psi_2
  }{%
    \Psi;\Gamma;\Sigma\vd e:t;\Psi_1
  }
\end{displaymath}

\draft{Mr. Oniszczuk gave two examples of ifs with x := new before if and
inside both branches of if the latter infers stronger type information since
it operates on two different type variable names }

\draft{ the old version
\begin{displaymath}
  \infer%
  {
    \Psi;\Gamma;\Sigma   \vd   \ite{e_1}{e_2}{e_3}
      : u;   \Psi_2\uplus\Psi'_2
  }%
  {
    \Psi;\Gamma;\Sigma   \vd   e_1
      : \cbool;\Psi_1
  &
    \Psi_1;\Gamma;\Sigma \vd   e_2
      : u;\Psi_2
  &
    \Psi_1;\Gamma;\Sigma \vd   e_3
      : u;\Psi'_2
  }
\end{displaymath}
}

\paragraph{Object structure introspection}

\newcommand{\ifhatrule}{\mbox{\small\rm (ifhasattr)}\xspace}
Dynamic languages split the processing not only depending on some condition
defined in terms of the actual values, but also depending on types of
expressions. Therefore, we introduce the construct that performs appropriate
check and provide a typing rule that handles it.

\begin{displaymath}
  \infer[\ifhatrule]%
    {
      \Psi;\Gamma,x:X;\Sigma   \vd   \iha{x}{a}{e_1}{e_2}
         : t;\Psi_2\uplus\Psi'_2
    }
    {
      \Psi[X\gets \{a\}];\Gamma,x:X;\Sigma   \vd   e_1
         : t;\Psi_2
    &
      \Psi;\Gamma,x:X;\Sigma   \vd   e_2
         : t;\Psi_2'
    }
\end{displaymath}

The typing rule must update the typing information available in the branches
of the instruction. Type information in one branch takes as granted that
the attribute $a$ is present. We need not pass the information that
the attribute $a$ is missing in the other branch since the case that the
actual value does not have the field $a$ must be expressed in the type
description by an alternative with $\bot$. 
\begin{itemize}
  \item $(\Psi,X\matches t)[X\gets \{a\}] = \Psi,X\matches (t[X\gets \{a\}])$,
  \item $\{ a:t_1\lor\cdots\lor t_n, \overline{l : u} \}%
     [X\gets\{a\}] =
     \{ a:t_{i_1}\lor\cdots\lor t_{i_k},\overline{l : u} \}$ where
     $\{t_{i_1},\ldots,t_{i_k}\} = \{ t_i \mid t_i\not=\bot \mbox{ and } t_i
     \in \{t_1,\ldots,t_n\}\}
     $.
\end{itemize}

We assume that base types are implemented so that they have a distinguished
field which makes possible to check for base type with $\iha{?}{?}{?}{?}$.


\paragraph{Variable and location access}

\newcommand{\vaccessrule}{\mbox{\small\rm (v-access)}\xspace}
\newcommand{\laccessrule}{\mbox{\small\rm (l-access)}\xspace}
The information about a type of a variable is recorded in the type environment
so we use it when a variable is referred in an expression.

\begin{displaymath}
  \infer[\vaccessrule]%
    {
      \Psi;\Gamma,x:t;\Sigma   \vd   x
        : t;\Psi
    }
    {}
\end{displaymath}

Similarly, the information about locations is stored in the location environment
$\Sigma$ and we exploit it in an analogous way. 

\begin{displaymath}
  \infer[\laccessrule]%
    {
      \Psi;\Gamma;\Sigma,l:t   \vd   l
        : t;\Psi
    }
    {}
\end{displaymath}

\paragraph{Function definition and call}

\newcommand{\fdeclrule}{\mbox{\small\rm (fdecl)}\xspace}
Whenever we want to type a function definition we must rely on the type
annotation that is associated with the function. We have to add the
type decoration to functions since they can be used in a recursive way.
The inference of a type in such circumstances must rely on some kind of
a fixpoint computation which is a difficult task. We assume here
that the type is given in the typechecking procedure either by hand
with help of programmer or by some kind of automatic type inference
algorithm.

The type we obtain for a function definition is just the type that
is explicitly given in the function. However, we have to check
in addition that the body of the function indeed obeys the declared type.
Therefore, we assume that the formal parameters have the types expressed
in the precondition of the function type and obey the constraints
noted in the constraint set from the same precondition. We expect then
that the resulting type is equal to the one in the postcondition of
the type and the constraints match the constraints in the postcondition.
It is worth noting that the constraints in the precondition and postcondition
should take into account not only the formal parameters, but also the
global variables that are visible in the function body. 

Notably, we type the function body so that the typing does not interfere with
the typing of the function definition. This reflects the dynamic character
of the language. In particular the running time type of a global variable may
differ at the function definition point from the expected one so that the
typing is impossible at that particular moment. However, the situation may
be different at the execution site. Therefore, we defer the check of the
compatibility of global variables to the function call site.

\begin{displaymath}
  \infer[\fdeclrule]%
    {
      \Psi;\Gamma;\Sigma   \vd   \kfunc(x_1,\cdots,x_n).t\; \{ e \}
        : t;\Psi
    }
    {
      \begin{array}{c}
      t \equiv \atype{\Psi_1}{t_1,\ldots, t_n}{u}{\Psi_2}
    \\
      \Psi_1;\Gamma,x_1:t_1,\ldots,x_n:t_n;\Sigma    \vd    e
        : u;\Psi_2
      \end{array}
    }
\end{displaymath}

In case a real application of a function expression is made to arguments
we have to check that the application indeed does not lead to type error in
this case. Therefore, the typing rule for application must check that
the type information for global variables at the call site (i.e. after
all actual parameters of the call are elaborated) is in accordance with the type
information available in the functional type. This is checked with
the relation $\vdc$.

\begin{definition}[$\vdc$]
The rules that operate on constraint sets are as follows. 
\begin{displaymath}
\begin{array}{c@{~~\qquad~~}c}
  \infer{\Psi\vdc\emptyset}{}
&
  \infer%
    {
      \Psi,X\matches t\vdc\Psi',X\matches t'
    }%
    {
       \Psi\vdc\Psi'
    &
       t\vdc t'
    }
\end{array}
\end{displaymath}

They reduce the problem to the one on record types where the relation is
defined as follows.

\begin{displaymath}
\begin{array}{c@{~\quad~}c@{~\quad~}c}
  \infer{t\vdc t}%
        {}
&
  \infer{\{\overline{l:t}\}\vdc\{\}}%
        {}
&
  \infer{\{a:u,\overline{l:t}\}\vdc\{\overline{l:t}\}}%
        {}
\\[2ex]
  \multicolumn{3}{c}{
  \infer{\{a:u,\overline{l:t}\}\vdc\{a:u',\overline{l:t}\}}
        { u \vdc u' }
  }
\\[2ex]
  u\vdc u
&
  u_1\vdc u_1\lor u_2
&
  u_2\vdc u_1\lor u_2
\\[2ex]
  \multicolumn{3}{c}{
  \displaystyle
  \infer{u_1\lor u_2\vdc u'_2\lor u'_1}%
        { u_1\vdc u'_1 &  u_2\vdc u'_2 }
   }
\\
\end{array}
\end{displaymath}

\draft{Maybe also add the following rule:
$$
  \infer{\{a:u,\overline{l:t}\}\vdc\{a:u\lor u',\,\overline{l:t'}\}}%
        {\{\overline{l:t}\}\vdc\{\overline{l:t'}\}}
$$
}
\end{definition}

\begin{proposition}
\label{thm:vdc-transitive}
  The relation $\vdc$ is reflexive and transitive.
\end{proposition}

\begin{proposition}
\label{thm:uplus-vdc}
  For any constraint sets $\Psi_1,\Psi_2$ we have $\Psi_1\vdc\Psi_1\uplus\Psi_2$.
\end{proposition}
\par\medskip

Execution of a function body causes changes of values held in global variables.
Such changes may give rise to changes in the types. This must by reflected
and the type information must be updated accordingly. Therefore, we need
an operation of such an update.

\begin{definition}[type information update $\updconstr$]
  \label{df:type-information-update}
  For two sets of constraints $\Psi_1, \Psi_2$ we define inductively the set
  of constraints $\Psi_1\updconstr\Psi_2$, in words
  {\em constraints $\Psi_1$ updated with $\Psi_2$}, as follows. 
  \begin{itemize}
    \item $\Psi\updconstr \emptyset = \Psi$,
    \item $\Psi,X\matches t\updconstr \Psi',X\matches t' =
           (\Psi\updconstr \Psi'),X\matches t'$,
    \item $\Psi\updconstr \Psi',X\matches t' =
           (\Psi\updconstr \Psi'),X\matches t'$
      when $X\not\in \dom{\Psi}$.
  \end{itemize}
\end{definition}

It is important to observe that the update preserves at least the
type information that is present in the constraint with which we
make the update.

\begin{lemma}
\label{thm:vdc-updconstr}
For any constraint sets $\Psi_1,\Psi_2$,
   $$\Psi_2\vdc\Psi_1\updconstr\Psi_2$$
\end{lemma}

\newcommand{\fapprule}{\mbox{\small\rm (fapp)}\xspace}
We can now present the rule for function application. This rule elaborates
first all the expressions in the actual parameters and updates the
constraints so that their side effects are taken into account and then
weakens the result to match the constraints for the input part of the
function type. Then the resulting set of constraint is just the set of
constraints after the last argument updated with the constraint information
contained in the result type of the function. 
\begin{displaymath}
  \infer[\fapprule]%
    {
      \Psi_1;\Gamma;\Sigma   \vd   e_f(N_1,\ldots,N_n)
        : u; \Psi_{n+1}\updconstr\Psi_2'
    }%
    {
      \begin{array}{c}
        \Psi_{i};\Gamma;\Sigma  \vd   N_i
          : t_i; \Psi_{i+1} \mbox{ for } i=1,\ldots,n
      \qquad
        \Psi_{n+1}\vdc\Psi_1'
      \\
        \Psi_1;\Gamma;\Sigma    \vd   e_f
          : \atype{\Psi_1'}{t_1,\ldots, t_n}{u}{\Psi'_2};\Psi_1
      \end{array}
    }
\end{displaymath}

\paragraph{Control flow with break instructions}

\newcommand{\labelrule}{\mbox{\small\rm (label)}\xspace}
\newcommand{\breakrule}{\mbox{\small\rm (break)}\xspace}
The language we propose includes non-local control flow instructions that
may be used to simulate exceptions or jumps out of loops (recursive calls
in our case). The labelled instruction resembles a call to an anonymous
function with no parameters. We only have to remember type information
that we want to achieve as the result of the execution. Note that
the prospective $\kbreak$ instructions can reside inside of
a recursive function call. Therefore, we cannot assume that the type we
give here explicitely is a result of a finitary collecting of a non-local
typing information from all the break statements inside $e$. 

\begin{displaymath}
  \infer[\labelrule]
    {
      \Psi;\Gamma;\Sigma   \vd   n.\atype{\emptyset}{}{t}{\Psi_2} \{ e \}
        : t;\Psi_2
    }
    {
      \Psi;\Gamma,n:\atype{\emptyset}{}{t}{\Psi_2};\Sigma   \vd   e
        : t;\Psi'_2
    && 
       \Psi'_2\vdc \Psi_2
    }
\end{displaymath}

In case we encounter $\kbreak$ instruction we have to check that
the expression the value of which it given as the result indeed obeys the
expected result type of the labelled instruction surrounding the break.

\begin{displaymath}
  \infer[\breakrule]
    {
      \Psi;\Gamma,n:\atype{\emptyset}{}{t}{\Psi_2};\Sigma   \vd  \breakle{n}{e}
        : t';\Psi'
    }
    {
      \Psi;\Gamma,n:\atype{\emptyset}{}{t}{\Psi_2};\Sigma   \vd  e
        : t;\Psi'_2
    &&
      \Psi'_2\vdc \Psi_2
    }
\end{displaymath}

\paragraph{Let-expression}

\newcommand{\letrule}{\mbox{\small\rm (let)}\xspace}
The $\klet$ expression is our language statement that does the instruction
sequencing operation. Therefore, we have to compute the constraints for
the first step, then for the second one and at last combine them together as
in the logical cut rule that forgets the middle formula.

\begin{displaymath}
  \infer[\letrule]%
    {
      \Psi_1;\Gamma;\Sigma   \vd    \letin{x}{e_1}{e_0}
         : t;\Psi_3
    }
    {
      \Psi_1;\Gamma;\Sigma   \vd    e_1
         : t_1;\Psi_2
    &&
      \Psi_2;\Gamma,x:t_1;\Sigma   \vd   e_0
         : t;\Psi_3
    }
\end{displaymath}

\paragraph{Object creation}

\newcommand{\newrule}{\mbox{\small\rm (new)}\xspace}
The object creation rule should just introduce
information that a new object is created. This is expressed by the fact
the set of constraints is extended with information on an object that
has no known to the type system fields.

\begin{displaymath}
  \infer[\newrule]%
    {
       \Psi;\Gamma;\Sigma   \vd  \knew
         : X;X\matches\{\},\Psi
     }
     {
        X\ \mbox{fresh}
     }
\end{displaymath}

\paragraph{Disjunction}

\newcommand{\weakenruleo}{\mbox{\small\rm (weaken-$\lor$)}\xspace}
We also add a rule to weaken the assumptions about
an expression. This is primarily necessary to obtain the subject reduction
property. In case an $\kif$ instruction is reduced to one of its branches
the type system infers only information from one of the branches and 
the information from the other one must be added artificially.
\begin{displaymath}
 \infer[\weakenruleo]%
    {
       \Psi_1;\Gamma;\Sigma   \vd  e: t_1\lor t_2;\Psi_2
     }
    {
       \Psi_1;\Gamma;\Sigma   \vd  e: t_i;\Psi_2 & i=1 \mbox{ or } 2
    }
\end{displaymath}

\section{Properties of the type system}
\label{sec:properties-of-the-type-system}

Soundness and subject reduction need to take the store into account. Hence we need to express the fact that a store instance satisfies certain constraints:

$$\sigma;\Sigma \models \Psi$$

where $\sigma$ is a store, 
$\Sigma$ a mapping between locations and types, and $\Psi$ a set of constraints of the form $X\matches t$.

\begin{definition}[$\models$]
\label{def:models}
If $c$ is a constant of type $t_c$:
$$\sigma;\Sigma\models c : t_c$$

For function values:

$$\infer[(\models \mathit{fun})]
{\sigma;\Sigma\models func(\vec{x}).t\{e\}:t}
{\Psi;\emptyset;\sigma\vd func(\vec{x}).t\{e\}:t;\Psi
&\sigma;\Sigma\models\Psi
}
$$

$$\sigma;\Sigma\models\emptyset$$

If $l$ is a location:
$$\infer[(\models obj)]{
  \sigma[l/\{\overline{a:v},\overline{y:s}\}];\Sigma(l:X)
  \models
  X\matches \{\overline{a:t}\},\Psi
  }
{\sigma;\Sigma\models \Psi 
& \sigma,\Sigma\models \overline{v:t}}
$$
in particular, for every object $o$:
$$\sigma(l/o);\Sigma(l:X)\models X\matches\{\} $$

For any $\Psi,\Psi'$
$$\infer[(\models \uplus)]{\sigma;\Sigma\models\Psi\uplus\Psi'}{\sigma;\Sigma\models\Psi}$$
\end{definition}

For any $t_1,t_2$,
$$ \infer[(\models\lor)]{\sigma;\Sigma\models v: t_1 \lor t_2}{\sigma;\Sigma\models v: t_i} $$

\begin{lemma}[soundness of $\vdc$]
\label{thm:vdc-sound}  
For any $\sigma,\Sigma$ and constraint sets $\Psi_1,\Psi_2$, 
if $\sigma,\Sigma\models \Psi_1$ and $\Psi_1\vdc\Psi_2$ then $\sigma,\Sigma\models \Psi_2$
\end{lemma}
\begin{proof}
It is sufficient to prove that if $\sigma,\Sigma\models X\matches o_1$ and $o_1\vdc o_2$ then $\sigma,\Sigma\models X\matches o_2$.
We proceed first by induction on the derivation of $\models$.
  The only nontrivial rules are 
  \begin{itemize}
  \item[($\models$ obj)] --- again, the only (marginally) nontrivial case is when  $o_1\vdc o_2$ was derived via the disjunction rule:
$$
  \infer{\{a:u,\overline{l:t}\}\vdc\{a:u\lor u',\,\overline{l:t'}\}}%
        {\{\overline{l:t}\}\vdc\{\overline{l:t'}\}}
$$
In this case we can derive $\sigma,\Sigma\models X\matches o_2$ by
inserting ($\models\lor$) before the ``$a:u$'' premise of the final
($\models$ obj) rule in the original derivation.

  \item[($\models\uplus$)] --- we have $\sigma,\Sigma\models X\matches o_1$, $o_2 = o_1 \uplus o'$, $o_2 \vdc o_3$ and need to prove  $\sigma,\Sigma\models X\matches o_3$, which can be done by induction on the definition of $o_2 \vdc o_3$.
  \end{itemize}
\end{proof}
\begin{lemma}
\label{thm:hasfield}
If   $$\sigma;\Sigma(l:X)\models X\matches\{a:t;\overline{f:s}\},\Psi'$$
 and $\vd t$ (i.e. $t$ contains no $\bot$), then 
 \begin{enumerate}[(i)]
 \item $\sigma(l)(a)$ is defined
 \item $\sigma;\Sigma\models\sigma(l)(a):t$
 \item $\sigma;\Sigma\models\Psi'$
 \end{enumerate}
\par\smallskip
\end{lemma}
\par\medskip
\begin{proof}
  By induction on the derivation of $\sigma;\Sigma(l:X)\models X\matches\{a:t;\overline{f:s}\}$.

If the last rule is $(\models obj)$ then the thesis follows immediately. On the other hand, if the last rule is $(\models\uplus)$, we have
$$X\matches\{a:t,\overline{f:u}\},\Psi' = \Psi_1\uplus\Psi_2$$
$$\infer[(\models \uplus)]{\sigma;\Sigma\models\Psi_1\uplus\Psi_2}{\sigma;\Sigma\models\Psi_i}$$
By the definition of $\uplus$, one of the following holds:
\begin{enumerate}[(a)]
\item  $\Psi_1=X\matches\{a:t,\overline{f:u}\},\Psi'_1,\;X\not\in dom(\Psi_2)$
\item  $\Psi_2=X\matches\{a:t,\overline{f:u}\},\Psi'_2,\;X\not\in dom(\Psi_1)$ 
\item $\Psi_1=X\matches\{\overline{l:s}\},\Psi'_1,\;\Psi_2=X\matches\{\overline{l':s'}\},\Psi'_2$
\end{enumerate}
If (a) or (b) holds, the thesis follows directly from the induction hypothesis. 

In the (c) case, by the assumption that $\vd t$, we have $a\in l,a\in l'$ as well as $t=t_1\lor t_2$ with $\vd t_1$ and $\vd t_2$.

Moreover, either
$$\sigma;\Sigma\models X\matches \{a:t_1,\overline{f:u}\},\Psi'_1$$
or
$$\sigma;\Sigma\models X\matches \{a:t_2,\overline{f':u'}\},\Psi'_1$$
Then by the induction hypothesis $\sigma(l)(a)$ is defined, and
$$ \sigma;\Sigma\models\Psi'_i$$
$$ \sigma;\Sigma\models\sigma(l)(a):t_i$$
for appropriate $i$. Thus by the rule $(\models\lor)$ we have
$$ \sigma;\Sigma\models\sigma(l)(a):t_1\lor t_2$$
which completes the proof. $\Box$
\end{proof}

\begin{lemma}[value completeness]
\label{thm:value-completeness}
 For any $\sigma,\Sigma,\Gamma,\Psi$ and value $v$ if $\sigma,\Sigma\models\Psi$ and $\sigma,\Sigma\models v:t$ then $\Psi;\Gamma;\Sigma\vd v:t;\Psi$
\end{lemma}
\begin{proof}
  By induction on the derivation of $\sigma,\Sigma\models v:t$. For the constant and function rules, the thesis follows directly.

 On the other hand if  the last rule used was 
$$ \infer[(\models\lor)]{\sigma;\Sigma\models v: t_1 \lor t_2}{\sigma;\Sigma\models v: t_i} $$
then by the induction hypothesis $\Psi;\Gamma;\Sigma\vd v:t_i;\Psi$, hence  $\Psi;\Gamma;\Sigma\vd v:t_1\lor t_2;\Psi$.
\end{proof}

\begin{lemma}[Progress]
\label{thm:progress}
If $\Psi_1;\Gamma;\Sigma\vd e:t;\Psi_2$ and 
$\sigma;\Sigma\models e:t; \Psi_1$
then either $e$ is a value, or 
$$\sigma,e \leadsto\sigma',e' $$
\end{lemma}

\begin{proof}
  The only expressions which may get stuck are record access (in case of lack of a field) and function application (in case of arity mismatch). Reducibility of record access follows from the Lemma~\ref{thm:hasfield}.
\end{proof}

\begin{lemma}[value stability]
\label{thm:stability}
If $v$ is a value and  
$$\Psi_1;\Sigma \vdash v:t;\Psi_2 $$

then  $\Psi_1\vdc \Psi_2$.

Moreover, for any $\Psi$ we have $\Psi;\Sigma \vdash v:t;\Psi $
\end{lemma}

\begin{proof}
  By induction on the derivation. The thesis follows directly for axioms and the function definition. For the disjunction rule use the induction hypothesis. For the $\uplus$ rule the thesis follows from Proposition~\ref{thm:uplus-vdc}.

 Since $v$ is a value, no other rules can be used in the derivation.
\end{proof}
\begin{lemma}[weakening]
\label{thm:weakening}
If $\Psi_1;\Gamma;\Sigma \vdash e:t;\Psi_2 $ and 
$\Gamma\subseteq\Gamma',\,\Sigma\subseteq\Sigma'$ then
$$\Psi_1;\Gamma;\Sigma' \vdash e:t;\Psi_2 $$
\end{lemma}
\begin{proof}
  By routine induction on derivation.
\end{proof}

\begin{lemma}[subject reduction]
\label{thm:subject}
Reduction preserves typings, i.e. if
$$\Psi_1;\Gamma;\Sigma \vdash e:t;\Psi_2 $$
$$\sigma;\Sigma \models \Psi_1$$
$$\sigma,\vec{E}\langle e\rangle\leadsto \sigma',\vec{E}\langle e'\rangle $$
then there exist $\Sigma',\,\Psi'$ such that  $\Sigma\subseteq\Sigma'$ and
$$\Psi';\Gamma;\Sigma \vdash e':t;\Psi_2 $$
$$\sigma';\Sigma' \models \Psi'$$
Moreover if $\sigma=\sigma'$ then the thesis holds for $\Sigma'=\Sigma$ and $\Psi'=\Psi$.
\end{lemma}
\begin{proof}
By induction on the derivation.  We now consider several cases depending on the last rule used

\paragraph{Case $e=\knew$:}
$$\Psi;\Gamma;\Sigma\vd \knew : X; X\matches\{\},\Psi$$
If  $\sigma;\Sigma\models \Psi_1$ then  
$$\sigma(l/\{\});\Sigma(l:X)\models X\matches\{\},\Psi$$

\paragraph{Case $e=l.a$:}
\begin{displaymath}
  \infer%
    {
      \Psi;\Gamma;l:X,\Sigma   \vd   l.a : u;\Psi
    }
    {
      \vd u
    &
      \Psi\ni X\matches\{a:u,\overline{f:t}\}
    }
\end{displaymath}
We have $\sigma,\econtext{l.x}\leadsto\sigma,\sigma(l)(a)$.
Since $\sigma;l:X,\Sigma\models\Psi$ then by Lemma~\ref{thm:hasfield}, $\sigma(l)$ has field $a$ of type $u$:
$$\sigma;\Sigma\models\sigma(l)(a):u$$
Hence by Lemma~\ref{thm:value-completeness}
$$ \Psi;\Gamma;l:X,\Sigma   \vd   \sigma(l)(a) : u;\Psi$$

\paragraph{Case $e=\ite{e_1}{e_2}{e_3}$:}
\par\smallskip
\begin{displaymath}
  \infer%
  {
    \Psi;\Gamma;\Sigma   \vd   \ite{e_1}{e_2}{e_3}
      : u;  \Psi_2 
  }%
  {
    \Psi;\Gamma;\Sigma   \vd   e_1
      : \cbool;\Psi_1
  &
    \Psi_1;\Gamma;\Sigma \vd   e_2
      : u;\Psi_2
  &
    \Psi_1;\Gamma;\Sigma \vd   e_3
      : u;\Psi_2
  }
\end{displaymath}

If $e_1$ is a value then $\sigma,\econtext{e}$ reduces to 
$\sigma,\econtext{e'}$ where $e'$ is either $e_2$ or $e_3$. 
In both cases we have
$$\Psi_1;\Gamma;\Sigma \vdash e':t;\Psi_2 $$
$$\sigma';\Sigma \models \Psi_2$$

Otherwise, $\sigma,\econtext{e_1}\leadsto\sigma',\econtext{e'_1}$. By induction hypothesis, there are
 $\Sigma',\,\Psi'$ such that
$$\Psi';\Gamma;\Sigma \vdash e_1':t;\Psi_2 $$
$$\sigma';\Sigma' \models \Psi'$$
By the typing rule for \textbf{if} and Lemma~\ref{thm:weakening}, we have
$$\Psi';\Gamma;\Sigma'\vd \ite{e_1}{e_2}{e_3} : t ;\Psi_2$$

\paragraph{Case $e=\kfunc(x_1,\cdots, x_n).t\; \{ e_b \}(v_1,\cdots,v_k)$}:
\begin{displaymath}
  \infer%
    {
      \Psi_1;\Gamma;\Sigma   \vd   
\kfunc(\vec{x}).t\; \{ e_b \}(\vec{v})
        : u; \Psi_{n+1}\updconstr\Psi_2'
    }%
    {
      \begin{array}{c}
        \infer{\Psi_1;\Gamma;\Sigma    \vd   \kfunc(\vec{x}).t\; \{ e_b \}
          : \atype{\Psi_1'}{t_1,\ldots, t_n}{u}{\Psi'_2};\Psi_1}
         {
\Psi'_1;\Gamma,\vec{x}:\vec{t};\Sigma\vd e_b:u;\Psi'_2
         }
      \\
        \Psi_{i};\Gamma;\Sigma  \vd   N_i
          : t_i; \Psi_{i+1} \mbox{ for } i=1,\ldots,n
      \qquad
        \Psi_{n+1}\vdc\Psi_1'
      \end{array}
    }
\end{displaymath}

For simplicity consider case when $n=1$:

\begin{displaymath}
  \infer%
    {
      \Psi_1;\Gamma;\Sigma   \vd   
\kfunc(x).t\; \{ e_b \}(v)
        : u; \Psi_{2}\updconstr\Psi_2'
    }%
    {
        \infer{\Psi_1;\Gamma;\Sigma    \vd   \kfunc(x).t\; \{ e_b \}
          : \atype{\Psi_1'}{t_1}{u}{\Psi'_2};\Psi_1}
         {
\Psi'_1;\Gamma,x:t_1;\Sigma\vd e_b:u;\Psi'_2
         }
       &
      \begin{array}{c}
        \Psi_{1};\Gamma;\Sigma  \vd   v_1
          : t_1; \Psi_{2} 
      \\
        \Psi_{2}\vdc\Psi_1'
      \end{array}
    }
\end{displaymath}

We have 
$$ \sigma,\econtext{\kfunc(x).t\; \{ e_b \}(v)} \leadsto \sigma,\econtext{(e_b[x/v])} $$
Consider an occurrence of a variable axiom for $x$ in the type derivation for $e_b$:
$$\Psi_x;\Gamma'\vd x : t_1;\Psi_x$$
By the stability lemma (\ref{thm:stability}), $\Psi_x;\Gamma'\vd v : t_1;\Psi_x$. Hence by replacing all such occurrences by appropriate type derivations for $v$, we get
$$\Psi'_1;\Gamma,x:t_1;\Sigma\vd e_b[x/v]:u;\Psi'_2$$

By the Lemma~\ref{thm:stability} $\Psi_1\vdc\Psi_2$, and $\Psi_2\vdc\Psi_1'$ is the side condition of the application. Hence by transitivity of $\vdc$ we have $\Psi_1\vdc\Psi_1'$ .

By Lemma~\ref{thm:vdc-updconstr}, we have $\Psi_2'\vdc \Psi_{2}\updconstr\Psi_2'$, hence
$$\Psi_1;\Gamma,x:t_1;\Sigma\vd e_b[x/v]:u;\Psi_{2}\updconstr\Psi'_2$$.
Which was to be proved.

\paragraph{Structural rules ($\uplus,\lor$)}
If the last rule used is
\begin{displaymath}
 \infer%
    {
       \Psi_1;\Gamma;\Sigma   \vd  e: t_1\lor t_2;\Psi_2
     }
    {
       \Psi_1;\Gamma;\Sigma   \vd  e: t_i;\Psi_2
    }
\end{displaymath}
then by the induction hypothesis the thesis holds for $t_i$, hence by reapplying the disjunction rule also for $t_1\lor t_2$.

If the last rule used is
\begin{displaymath}
\infer{  \Psi;\Gamma;\Sigma\vd e:t;\Psi_2\uplus\Psi_3}%
{\Psi;\Gamma;\Sigma\vd e:t;\Psi_2}
\end{displaymath}
Then by the induction hypothesis 
$$\Psi';\Gamma;\Sigma \vdash e':t;\Psi_2 $$
$$\sigma';\Sigma' \models \Psi'$$
\end{proof}
And we can get the thesis by applying the $\uplus$ rule.

\begin{theorem}[soundness]
If
  $\Psi_1;\Gamma;\Sigma \vdash |\vec{E}\langle e\rangle|:t;\Psi_2 $
and
$$\sigma;\Sigma\models \Psi_1$$

then either $\sigma, \vec{E}\langle e\rangle$ has an infinite reduction path, or it reduces to a value $v$, i.e. there exist $\sigma',\Sigma'$ and value $v$
such that
  $$\sigma, \vec{E}\langle e\rangle\leadsto^{*} \sigma', v$$
and 
$$\sigma';\Sigma'\models v:t,\Psi_2.$$
\end{theorem}
\begin{proof}

  By the progress lemma (\ref{thm:progress}), the reduction can stop only at a value. Required properties of the final value and state follow from lemmas~\ref{thm:subject} (by induction on the lreduction length) and~\ref{thm:stability}.
\end{proof}

\section{Related work}
\label{sec:related-work}

The starting point of our research is the paper by Guha et al \cite{GuhaSK11}
where a type system for scripting languages is presented. The type system
there does not address the types of objects and only infers the typing
information concerning base types (including special type of references to
take into account object in a very simple way). We provide here a type
system which in addition to the one by Guha et al can infer meaningful
typing information for objects. 

In addition, the paper by Guha et al relies on runtime tags which are present
in the semantics of dynamicly typed languages such as JavaScript. The
checks over runtime type tags can be viewed as asserts that check if the
particular value is of an expected type. This primitive has, however, one
disadvantage. Namely, it does not reflect the split in the control flow
that makes possible to use particular type. Still, the scripting languages
have operators that check for the actual running time type of a value (e.g.
{\sf typeof} function and {\sf in} operator in JavaScript, or {\sf hasattr},
{\sf getattr}, and {\sf getattr\_static} in Python). We consider it more
natural to rely on these operators instead of the tagchecks.

Our typing framework works in a fashion such that, in principle, all the
typing information should be present. This scenario is not realistic 
when the expectations of the scripting languages programmers are taken into
account. Each real effort to bring static information into the program
development here must consider what was explicated in \cite{WrigstadEFNV09}
and realised in a scpecific way in {\sf Thorn} scripting language
\cite{BloomFNORSVW09} and led to interesting type architecture
\cite{WrigstadNLOV10a,WrigstadNLOV10b}. 

\section{Conclusions and future work}
\label{sec:conclusions}

The type system presented in this paper gives an expressible framework for
typing objects in dynamic languages such as Python or JavaScript. It follows
the general view of type-and-effect systems where running time program
property must be expressed as a dynamical change. We express this change
by means of Hoare-like triples that describe the structure of relevant
objects before an expression is executed and after the execution.

The type system is presented in a fashion similar to Church-style. The
necessary typing information must be provided by a programmer. We work now
on helpful heuristics that will infer most of the type information to make
the annotational effort of a programmer minimal and on the way a system
with types can be integrated to real program development in scripting languages.

\bibliographystyle{alpha}
\bibliography{refl}

\end{document}
\endinput

$$\infer{}{}
$$